\documentclass[preprintnumbers,11pt,onecolumn]{article}
\pdfoutput=1

\usepackage{fullpage}
\usepackage{amsfonts, amssymb, amsmath, amsthm}
\usepackage{latexsym}
\usepackage[tracking=smallcaps]{microtype}	
\usepackage{url}
\usepackage{color}
\definecolor{DarkGray}{rgb}{0.1,0.1,0.5}
\usepackage[colorlinks=true,breaklinks, linkcolor=black,citecolor=black,urlcolor=black]{hyperref}	

\usepackage{graphicx}
\usepackage[tight, TABBOTCAP]{subfigure}

\newcommand{\bra}[1]{{\langle#1|}}
\newcommand{\ket}[1]{{|#1\rangle}}
\newcommand{\braket}[2]{{\langle#1|#2\rangle}}
\newcommand{\ketbra}[2]{{\ket{#1}\!\bra{#2}}}

\newcommand{\abs}[1]{{\lvert #1\rvert}}	

\newcommand{\binomial}[2]{\ensuremath{\left(\begin{smallmatrix}#1 \\ #2 \end{smallmatrix}\right)}}

\DeclareMathOperator{\Tr}{\operatorname{Tr}}

\def\H {{\mathcal H}}

\newcounter{sprows}

\newlength{\spheight}
\newlength{\spraise}

\newlength{\commentslength}

\newcommand{\rem}[1]{}

\newtheorem{theorem}{Theorem}[section]
\newtheorem{lemma}[theorem]{Lemma}
\newtheorem{corollary}[theorem]{Corollary}

\newfont{\subsubsecfnt}{ptmri8t at 11pt}
\renewcommand{\subparagraph}[1]{\smallskip{\subsubsecfnt #1.}}

\newcommand{\eqnref}[1]{\hyperref[#1]{{(\ref*{#1})}}}
\newcommand{\thmref}[1]{\hyperref[#1]{{Theorem~\ref*{#1}}}}
\newcommand{\lemref}[1]{\hyperref[#1]{{Lemma~\ref*{#1}}}}
\newcommand{\corref}[1]{\hyperref[#1]{{Corollary~\ref*{#1}}}}
\newcommand{\defref}[1]{\hyperref[#1]{{Definition~\ref*{#1}}}}
\newcommand{\secref}[1]{\hyperref[#1]{{Section~\ref*{#1}}}}
\newcommand{\figref}[1]{\hyperref[#1]{{Figure~\ref*{#1}}}}
\newcommand{\tabref}[1]{\hyperref[#1]{{Table~\ref*{#1}}}}
\newcommand{\remref}[1]{\hyperref[#1]{{Remark~\ref*{#1}}}}
\newcommand{\appref}[1]{\hyperref[#1]{{Appendix~\ref*{#1}}}}
\newcommand{\claimref}[1]{\hyperref[#1]{{Claim~\ref*{#1}}}}
\newcommand{\factref}[1]{\hyperref[#1]{{Fact~\ref*{#1}}}}
\newcommand{\propref}[1]{\hyperref[#1]{{Proposition~\ref*{#1}}}}
\newcommand{\exampleref}[1]{\hyperref[#1]{{Example~\ref*{#1}}}}
\newcommand{\conjref}[1]{\hyperref[#1]{{Conjecture~\ref*{#1}}}}

\allowdisplaybreaks[1]

\def\COLOR{}
\ifdefined\COLOR

\else

\fi

\usepackage{array}
\usepackage{scalerel}  

\def\beq{\begin{equation}}
\def\eeq{\end{equation}}

\graphicspath{{./images/}}

\begin{document}
\def\compilefullpaper{}

\title{Quantum dimension test using the uncertainty principle}
\author{Rui Chao \qquad Ben W. Reichardt \\ University of Southern California}
\date{}
\maketitle

\begin{abstract}
We propose a test for certifying the dimension of a quantum system: store in it a random $n$-bit string, in either the computational or the Hadamard basis, and later check that the string can be mostly recovered.  
The protocol tolerates noise, and the verifier only needs to prepare one-qubit states.  
The analysis is based on uncertainty relations in the presence of quantum memory, due to Berta et al.~(2010).  
\end{abstract}

\section{Introduction}

The dimension of a quantum  system is the maximum number of perfectly distinguishable states that can be stored, alone or in superposition.  
Quantum dimension is a fundamental characteristic of a system.  
It is crucial for quantum computation, and for quantum cryptography.  
It is useful to take a conservative perspective with minimal physical assumptions, and then develop practical methods for certifying the dimensionality of an uncharacterized quantum system.  

Quantum dimension is different from, and superior to, classical dimension. 
Although a classical system of $n$ bits is $2^n$ dimensional, the $2^n$ parameters of an $n$-qubit system can interfere with each other. 
Interference and entanglement allow for 
the numerous algorithmic speedups and information-theoretic advantages of quantum computers. 

Here, we propose a fairly simple test for the dimension of a quantum system.  
The verifier encodes a uniformly random $n$-bit string into $n$ qubits, in either the computational or the Hadamard basis, and then sends the qubits to the device. (Operationally, she can prepare and send one qubit at a time.) 
Later, she reveals the basis, and asks for the string back.  
She compares it to the original string, and accepts if they agree on at least $(1 - \alpha) n$ bits. 
Crucially, allowing $\alpha > 0$ makes the protocol resilient to noise. 

To prove the validity of the test, we show that it can be fitted into the very useful framework of \emph{uncertainty principle with quantum memory}~\cite{Berta10uncertainty,Coles17rmp}.  
This principle generalizes the traditional Heisenberg uncertainty principle, in the sense that the measured system can be entangled with a quantum memory, and the corresponding uncertainties, quantified as entropies, are conditioned on accessing the information from the memory. 
As we will formulate in \secref{sec:protocolentanglement}, there is a straightforward equivalence between the bipartite entangled setting and our protocol.  
As a result, our analysis reduces to estimating the uncertainties bounded, using the principle, by the amount of entanglement, which further bounds the dimension. 

\medskip

\secref{sec:results} formally describes our protocol, and states the certifiable quantum dimension lower bound (\thmref{thm:bound}). 
\secref{sec:proof} proves the theorem. 
\secref{sec:numerics} gives numerics under a simplified noise model.  

\subsection*{Related work}

Formulating rigorous notions of dimension in general physical theories is a nontrivial problem~\cite{barrett2007information,brunner2014dimension}.
Here we restrict ourselves to quantum dimension, i.e., the dimension of the underlying Hilbert space of the system. 
Below we briefly review previous tests along this line, which involve fairly diverse assumptions and applicable contexts. 
Most of them essentially certify high entanglement, with dimension lower bounds following as corollaries. 

The first formal test is the ``dimension witness''~\cite{Brunner08dimension}, which is a Bell-type inequality fulfilled by all correlations generated by local quantum measurements on multipartite states within fixed dimension.  Dimension witness is then made ``device independent''~\cite{Gallego10didimension}, where the conditional probabilities in the inequalities are generated from an ensemble of single-party states, rather than a single multipartite state. 
This framework has been extended to a realistic scenario~\cite{dall12robustdidimension}, and related to quantum random access codes~\cite{Wehner08dimension} and state discrimination~\cite{Brunner13dimension}. 
While efficient tools~\cite{navascues2015bounding, sikora2016minimum} exist for computing a witness for arbitrary dimensions, the systems in experiments so far~\cite{hendrych12experimentdidimension, Ahrens12experimentdidimension, Ahrens14dimensionexperiment, Ambrosio14didimension, cai2016new, sun2016experimental} have not exceeded dimension 13.

Another closely related but stronger method, still based on Bell-type inequalities, is ``self testing.'' 
It states that a nearly optimal strategy for a specific non-local game implies the closeness of the underlying quantum system to a certain entangled state, thus lower bounding the dimension of individual local systems. 
Protocols exist for robust self-testing of asymptotically many EPR states~\cite{McKague15parallelselftesting, coudron2016parallel, coladangelo2017parallel, chao2016test, overlapping17csrv, natarajan2017quantum, natarajanlow}, or even infinite entanglement~\cite{slofstra2018group,ji2018three,coladangelo2019two}. 
The idea of self-testing can also be extended to the certification of one-shot distillable entanglement~\cite{arnon2019device} and entanglement of formation~\cite{arnon2017noise}. 
Recent experiments~\cite{zhang2018experimental, wang2018multidimensional} self-test up to 15-dimensional  bipartite entangled states on photonic platforms.

Besides Bell-type inequalities, another common technique is tomography.  
Full state tomography on a $d$-dimensional bipartite state requires measurement of $O(d^2)$ product bases,  while~\cite{Bavaresco18twomeasurement} shows that only two bases suffice for adaptive measurements.  
Their experiment is able to certify 9-dimensional entanglement within an 11-dimensional entangled photon pair.  

In experiments based on the above methods, typically the first step is to prepare a bipartite system and both subsystems are supposed to support high-dimensional quantum storage.  
In this work, we present a protocol based on the uncertainty principle with quantum memory, in which the verifier in principle only needs one qubit in storage at a time.  
The uncertainty principle with quantum memory has been investigated in experiments~\cite{li11uncertaintyexperiment, prevedel11uncertaintyexperiment} where individual photon pairs are repeatedly generated and measured locally in the Pauli $X$ or $Z$ basis.  
In particular, the quantum dimension that is effectively certified is up to two.  
Here we extend their analysis so that our protocol is noise-tolerant and can certify asymptotically large quantum dimension.  

Pfister et al.~\cite{pfister2018capacity} have proposed and implemented a test that is similar to ours, but with crucial distinctions. 
In their test, for instance, the verifier's random string and basis, once chosen, are immediately communicated to the prover. 
More importantly, their test certifies a lower bound of \emph{one-shot quantum capacity}~\cite{buscemi2010quantum,tomamichel2016quantum}, which is a different characterization of quantum storage than quantum dimension. 

\section{Dimension lower bound}
\label{sec:results}
 
\begin{figure}[!t]
{ \noindent \hrulefill \\  
\centering \textbf{Dimension test protocol} \\ } \medskip
Let $n \geq 1$ be an integer and $\alpha \in [0, 1/2)$.  
\begin{enumerate}
\item Alice chooses $\Theta \in \{0,1\}$ and $S \in \{0,1\}^n$, both uniformly at random.  \\
For $j = 1, \ldots, n$, she prepares and sends to Bob a qubit in state $H^\Theta \ket{S_j}$, where $H = \tfrac{1}{\sqrt 2} \big(\begin{smallmatrix}1&1\\1&\!-1\end{smallmatrix}\big)$ is the Hadamard matrix.  
\item Alice announces her basis choice $\Theta$ to Bob.  
\item Based on $\Theta$, Bob measures his system, and outputs a guess~$S' \in \{0,1\}^n$.  
\end{enumerate}
Bob passes the test if and only if $S$ and $S'$ match on at least $(1-\alpha) n$ bits.  \\
\vspace{0\baselineskip}
\hrulefill
\caption{A test for certifying the quantum dimension of Bob's system.} \label{fig:protocol}
\end{figure}

\figref{fig:protocol} outlines the protocol.  Essentially, Alice sends to Bob a random string, in either the computational or the Hadamard basis, and later she checks that it can mostly be recovered.  Observe that Alice only needs to be able to prepare and transmit one qubit at a time.  By setting the parameter $\alpha$ appropriately, a certain amount of noise on the qubits can be tolerated.  We discuss this further in \secref{sec:numerics}.  

\figref{fig:protocol} shows all of Alice's actions.  Bob's goal is to pass the test.  What can Bob do?  For the test to have any validity, we assume that Alice's choices of $S$ and~$\Theta$ are secret from Bob, the latter until she announces it.  In particular, there should not be a side channel by which Bob can learn this information.  Aside from that, however, we try to avoid limiting Bob's ability.  In particular, we allow Bob arbitrary computational power, and arbitrary classical memory.  Assume that Bob has two registers, $C$ and~$Q$, classical and quantum, respectively.  The registers can be initialized arbitrarily, independent of $S$ and~$\Theta$.  One can imagine various dynamics for how Bob processes each qubit as it arrives from Alice.  This is not important.  Conservatively, we can model Bob as receiving $(H^{\otimes n})^\Theta \ket{S}$ all at once.  Then at the end of step 1, we allow Bob to apply an arbitrary classical-quantum channel $\mathcal C_B : ({\mathbb C}^2)^{\otimes n} \rightarrow \H_B = \H_C \otimes \H_Q$.  
In step~$3$, Bob can make arbitrary manipulations to determine his response~$S'$; without loss of generality, they can be taken to be a POVM on $\H_C \otimes \H_Q$.  

Our main theorem is: 

\begin{theorem} \label{thm:bound}
In the test of \figref{fig:protocol} with $\alpha \leq 1/2$, if Bob passes with probability $p$, then
\begin{equation}
\log \dim \H_Q 
\geq n - 2 H(p) - 2 p \log M - 2 (1-p) \log (2^n - M)
 \enspace ,
\end{equation}
where $M = \sum_{i=0}^{\alpha n} \binom{n}{i}$ and $H(x) = -x \log x - (1-x) \log (1-x)$ is the binary entropy function.  In particular, 
\begin{equation} \label{eqn:stirlingbound}
\dim \H_Q
\geq 
2^{\vstretch{1.15}{(} (1 - H(\alpha)) 2 p - 1 \vstretch{1.15}{)} \, n - 2 H(p)}
 \enspace .
\end{equation}
\end{theorem}

In an experiment, of course, the test needs to be repeated to get an empirical estimate $\hat p$ of~$p$; and the dimension lower bound holds with statistical confidence.  

\section{Proof of \thmref{thm:bound}}
\label{sec:proof}

\subsection{Entanglement-based protocol}
\label{sec:protocolentanglement}

The test in \figref{fig:protocol} is mathematically equivalent to one in which Alice and Bob initially share $n$ EPR states, and Alice measures her qubits in either the Pauli $Z$ or $X$ basis to generate her string~$S$; see \figref{fig:protocolentanglement}.  
(In both cases, the joint distribution of  Alice's recorded string and the corresponding received or marginal state  on Bob's side are the same.)  
The protocol in \figref{fig:protocol} is easier to implement, since Alice needs to store only one qubit at a time, but the protocol in \figref{fig:protocolentanglement} can be analyzed using entanglement measures.  

\begin{figure}[!t]
{ \noindent \hrulefill \\  
\centering \textbf{Entanglement-based dimension test protocol} \\ } \medskip
Let $n \geq 1$ be an integer and $\alpha \in [0, 1/2)$.  Assume that Alice and Bob initially share halves of $n$ EPR states, $\tfrac{1}{2^{n/2}} ( \ket{00} + \ket{11} )^{\otimes n}$.  
\begin{enumerate}
\item 
Bob applies channel $\mathcal C_B$ to his $n$ qubits, with a classical-quantum output.
\item 
Alice chooses $\Theta \in \{X, Z\}$ uniformly, and announces it to Bob.  
\item 
Alice measures her $n$ qubits in the Pauli $\Theta$ basis, obtaining string $S \in \{0,1\}^n$.  \\
Bob measures his system with a POVM and returns the outcome $S' \in \{0,1\}^n$. 
\end{enumerate}
Bob passes the test if and only if $S$ and $S'$ match on at least $(1-\alpha) n$ bits.  \\
\vspace{0\baselineskip}
\hrulefill
\caption{A mathematically equivalent dimension test.} \label{fig:protocolentanglement}
\end{figure}

Denote the two parties' joint state, before any measurements, as $\rho_{AB} = \frac{1}{2^n} \left(\mathcal I_A \otimes \mathcal C_B \right) \big[(\ket{00} + \ket{11})(\bra{00} + \bra{11}) \big]^{\otimes n} \in \H_A \otimes \H_B$.  
The dimension of Bob's Hilbert space $\H_B$ can be related to the entanglement of the state~$\rho_{AB}$.  
To this end, we need the notions of von Neumann entropy $H(\sigma) := -\Tr[\sigma\log\sigma]$, and conditional von Neumann entropy $H(A \vert B)_\sigma := H(\sigma_{AB}) - H(\Tr_A \sigma_{AB})$.  
$H(\{p_i\})$ denotes the Shannon entropy of discrete probability distribution $\{p_i\}$, and all logarithms are base 2.  

From the definition of conditional von Neumann entropy it follows that
\begin{equation} \label{eqn:logdim}
\log(\min\{\dim \H_A, \dim \H_B\}) \geq -H(A \vert B)_\rho \, . 
\end{equation}
The equality holds if and only if $\rho_{AB}$ is a maximally entangled state. 
In fact, $H(A \vert B)$ serves as a witness to a closely related measure of entanglement, titled \emph{distillable entanglement}~\cite{devetak2005distillation}, via the relation $E_{\text{distill}}(\rho) \geq - H(A \vert B)_\rho$.  

Furthermore, $-H(A \vert B)_\rho$ lower bounds the dimension of Bob's quantum register: 
\begin{equation} \label{eqn:logdimquantum}
\log \dim \H_Q \geq -H(A \vert B)_\rho \, .
\end{equation}
Indeed, suppose the classical-quantum output of Bob's channel is given by
\begin{equation*}
\rho_B = \Tr_A \rho_{AB} = \sum_i p_i \ketbra i i \otimes \sigma_i \in \H_C \otimes \H_Q
\enspace .
\end{equation*}
Then the full state $\rho_{AB}$ must have the form $\rho_{AB} = \sum_i p_i (\ketbra i i)_C \otimes (\sigma_i')_{AQ}$, where $\Tr_A \sigma_i' = \sigma_i$.  
So 
\begin{equation*}\begin{split}
H(A \vert B)_\rho 
&= H(\rho_{AB}) - H(\rho_B) 
= \Big[ H(\{p_i\}) + \sum_i p_i H(\sigma'_i) \Big] - \Big[H(\{p_i\}) + \sum_i p_i H(\sigma_i) \Big] \\
&\geq \min_i \big( H(\sigma'_i) - H(\sigma_i) \big) 
= \min_i H(A \vert Q)_{\sigma'}
\geq - \log \dim \H_Q
 \, .
\end{split}\end{equation*}

\subsection{Uncertainty principle with quantum memory}
\label{sec:uncertainty}

In this section, we introduce an uncertainty principle with quantum memory, which allows upper bounding $H(A \vert B)_\rho$ by Bob's uncertainty on Alice's measurement outcomes.  
In~\secref{sec:bounds}, we will show how to upper bound Bob's uncertainty through only one parameter, namely, the probability of passing the dimension test.  

There are several versions of uncertainty relations; see~\cite{Coles17rmp} for a review of their history and applications.  
We will use an uncertainty relation in the presence of quantum memory, and in terms of von Neumann entropy,  developed by~\cite{Christandl2005uncertainty,Renes2009conjectured,Berta10uncertainty}. 
Incorporating memory is important for our two-party scenario, and von Neumann entropy is simpler to manipulate and empirically estimate, and is more readily related to system dimension (see~\eqnref{eqn:logdim}), compared to other types of entropies.  

\begin{theorem}[\cite{Berta10uncertainty}]
\label{thm:uncertainty}
For a quantum state $\rho \in \H_A \otimes \H_B$, and orthonormal bases $X = \{ \ket x \}$ and $Z = \{ \ket z \}$ for $\H_A$, 
\begin{equation} \label{eqn:uncertainty}
H(A \vert B)_{X(\rho)} + H(A \vert B)_{Z(\rho)} - \log \tfrac1c \geq H(A\vert B)_\rho \, , 
\end{equation}
where $c = \max_{x,z} \abs{\braket x z}^2$, and $X(\rho) = \sum_x \ketbra x x_A \,\rho \, \ketbra x x_A$ and $Z(\rho) = \sum_z \ketbra z z_A \,\rho \, \ketbra z z_A$ are the results of measuring $\rho$ in the $X$ and $Z$ bases, respectively.  
\end{theorem}

Eq.~\eqnref{eqn:uncertainty} lower-bounds Bob's uncertainty on Alice's measurement outcomes in terms of the complementarity~$c$ of Alice's measurements and of the entanglement in the initial state~$\rho$.  
For example, if $X$ and $Z$ are mutually unbiased bases, then $c = 1 / \dim \H_A$, and by Eq.~\eqnref{eqn:logdim}, $\log \dim \H_B \geq -H(A \vert B)_\rho \geq \log \dim \H_A - \big( H(A \vert B)_{X(\rho)} + H(A \vert B)_{Z(\rho)} \big)$.  Therefore an upper bound on the sum of the entropies of Alice's measurement outcomes conditioned on Bob's system implies a lower bound on $\dim \H_B$.  

\thmref{thm:uncertainty} has been demonstrated in proof-of-principle experiments on single pairs of photons in~\cite{li11uncertaintyexperiment, prevedel11uncertaintyexperiment}.  
These experiments effectively certify quantum dimension of up to two.  
They estimate and upper-bound $H(A \vert B)_{X(\rho)} + H(A \vert B)_{Z(\rho)}$ using techniques such as the data-processing inequality and Fano's inequality.  
In the next section, we generalize these techniques so that an asymptotically large quantum dimension can be certified, in a noise-tolerant way.

\subsection{Upper bounds of conditional entropies}
\label{sec:bounds}

For analyzing the protocol in \figref{fig:protocolentanglement}, denote  
\begin{equation*}\begin{split} 
H(S \vert S', \Theta) 
&= \frac12 \Big( H(S \vert S', \Theta=X) + H(S \vert S', \Theta=Z) \Big) \\
H(S \vert B, \Theta) 
&= \frac12 \Big( H(S \vert B, \Theta=X) + H(S \vert B, \Theta=Z) \Big) \\
&= \frac12 \Big( H(A \vert B)_{X(\rho)} + H(A \vert B)_{Z(\rho)} \Big)
 \enspace . 
\end{split}\end{equation*}
Let $p^\theta$ be the probability that Bob passes the test conditioned on Alice measuring in basis~$\theta$, thus $p = \tfrac12 (p^X + p^Z)$.  

\begin{lemma}[Data-processing inequality] \label{lem:data}
For $\theta \in \{X, Z\}$, 
\begin{equation*}
H(S \vert B, \Theta = \theta) \leq H(S \vert S', \Theta = \theta)
 \enspace .  
\end{equation*}
\end{lemma}

\begin{proof}
Bob obtains $S'$ by measuring his register~$B$.  This is a quantum channel.  
According to the data processing inequality (or more precisely strong subadditivity of quantum entropies~\cite{lieb1973fundamental, lieb1973proof}), the conditional von Neumann entropy is non-decreasing after applying a quantum channel on the conditioned system.  
\end{proof}

\begin{proof}[Proof of~\thmref{thm:bound}]
Since the complementarity of Alice's measurements is $c = 1 / 2^n$, we have 
\begin{align*}
\log \dim \H_Q 
&\geq -H(A \vert B)_\rho && \text{by Eq.~\eqnref{eqn:logdimquantum}} \\
&\geq n - 2 H(S \vert B, \Theta) && \text{(\thmref{thm:uncertainty})} \\
&\geq n - 2 H(S \vert S', \Theta) && \text{(\lemref{lem:data})}
\end{align*}

Let $E = 1$ if Bob passes the protocol, i.e., $\abs{S \oplus S'} \leq \alpha n$, and $E = 0$ otherwise.  Applying the chain rule to $H(E S \vert S' \Theta)$ in two ways, 
\begin{equation*}\begin{split}
H(E S \vert S' \Theta) 
&= H(S \vert S'\Theta) + H(E \vert SS'\Theta) \\
&= H(E \vert S'\Theta) + H(S\vert ES'\Theta) 
 \enspace .
\end{split}\end{equation*}
Note that $H(E \vert S S' \Theta) = 0$ and $H(E \vert S'\Theta) \leq H(E \vert \Theta) = \frac12 \big( H(p^X) + H(p^Z) \big) \leq H(p)$.  
Thus we obtain a version of Fano's inequality, 
\begin{align*}
H(S \vert S', \Theta)
&\leq H(p) + p H(S \vert E=1, S', \Theta) + (1-p) H(S \vert E=0, S', \Theta)\\
&\leq H(p) + p \log M + (1 - p) \log (2^n - M) \,, 
\end{align*}
where $M = \sum_{i=0}^{\alpha n} \binomial{n}{i}$ is the number of possible winning choices for $S$, given $S'$.  

It remains just to estimate~$M$.  By~\cite[Lemma~16.19]{FlumGrohe06stirling}, using Stirling's approximation, 
\begin{equation*}
2^{n H(\alpha)} \geq 
M 
\geq \binom{n}{\alpha n}
\geq \frac{1}{\sqrt{8 n \alpha (1 - \alpha)}} 2^{n H(\alpha)} 
 \enspace .
\end{equation*}
Substituting $0 < M \leq 2^{n H(\alpha)}$ into the uncertainty relation gives Eq.~\eqnref{eqn:stirlingbound}.  
\end{proof}

\section{Performance in small quantum devices}
\label{sec:numerics}

To illustrate the practicality of our protocol, we calculate the number of qubits that can be certified using a noisy, intermediate-scale quantum (NISQ) device~\cite{Preskill18nisq} with $n \leq 90$ qubits.  To be concrete, we assume an honest Bob who always measures the received qubits in the same basis as Alice, which is optimal. In addition, we assume the following simplified noise model for simulating the protocol.  

\begin{enumerate}
\item Alice initializes qubits always in $\rho = \ketbra 0 0$, afflicted by a bit-flip channel $(1 - p_1) \rho + p_1 X \rho X$.
\item To encode a bit in certain basis, Alice applies a proper unitary rotation to her qubit, followed by a depolarizing channel $(1 - \frac43 p_2)\rho + \frac23 p_2 I$. For example, Alice applies $H X$ when encoding~$\ket - $, whereas she does nothing when encoding $\ket 0$.
\item Alice sends each qubit to Bob through a dephasing channel $(1 - p_3) \rho + p_3 Z \rho Z$.
\item Bob can only measure in the Pauli $Z$ basis, and the outcome is flipped with probability $p_4$.
\item To measure a qubit in basis other than Pauli $Z$, Bob can apply, before measuring $Z$, a proper unitary rotation,  which is followed by a depolarizing channel $(1 - \frac43 p_2) \rho + \frac23 p_2 I$.
\end{enumerate}

Errors of operations on different qubits or in different time-steps are independent. 
We choose the error rates $p_1, \ldots, p_4$ to be proportional to the infidelities of qubit reset, single-qubit gate, qubit shuttling and measurement, respectively, whose values are  from~\cite{bermudez2017assessing,kaufmann2018high}: 
\begin{equation*}
p_1 \propto 5 \cdot 10^{-3} \qquad p_2 \propto 5 \cdot 10^{-5} \qquad p_3 \propto 6 \cdot 10^{-6} \qquad p_4 \propto 10^{-3} \enspace .
\end{equation*}

\begin{figure}
\centering  
\includegraphics[scale=.5]{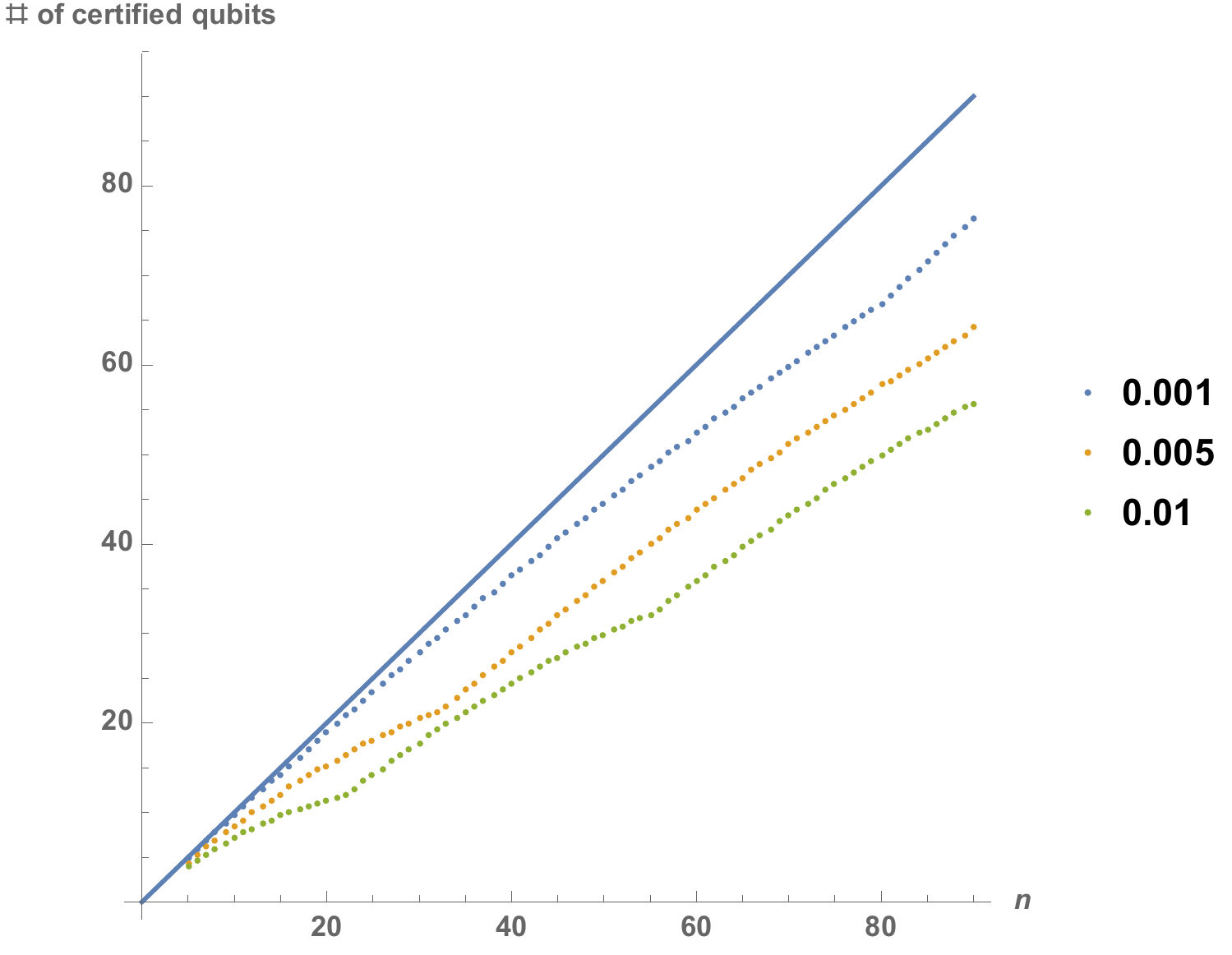}
\caption{Number of certified qubits out of $n$. 
For different values of the total noise rate, $\sum_{j=1}^4 p_j \in \{.001,.005,.01\}$, we plot the maximum number of certified qubits optimized over threshold~$\alpha$, for $5 \leq n \leq 90$.   
The solid blue line is $y = x$.  
} \label{fig:plot} 
\end{figure}

In our analysis, for fixed $n$ and threshold $\alpha$, the certified dimension is determined by  the winning probabilities $p^\theta$s, which are further determined by $p_1, \ldots, p_4$. For specific $n$ and total error rate, we calculate and plot the maximum number of certified qubits optimized over possible values of the threshold $\alpha$. See~\figref{fig:plot}. Note that the numerics do not consider the statistics to estimate the $p^\theta$s required in practice.

\section{Conclusion}
\label{sec:conclusion}

We have proposed a test for certifying the dimension of a quantum system, and analyzed it using the entropic uncertainty principle with quantum memory.  
Compared with previous schemes, our protocol is conceptually simpler and easier to implement.  
We envision that the test can soon be used in experiments to certify tens of qubits. 
One challenging problem is to determine the optimal cheating strategy for a quantum device with limited dimension.  
It is also interesting to ask whether the optimal strategy is unique.

\section*{Acknowledgements}

This work is supported by NSF grant CCF-1254119 and ARO grant W911NF-12-1-0541, and MURI Grant FA9550-18-1-0161.

\appendix

\section{Uncertainty relations for multiple measurements}
\label{sec:multiple}

The dimension test in~\figref{fig:protocol} is based on the uncertainty relation~\thmref{thm:uncertainty} where Alice measures the $n$ qubits in  either the Pauli $X$ or $Z$ basis transversally. 
However, uncertainty relations with more than two measurement bases abound; see~\cite{Coles17rmp}. 
It is then natural to generalize our protocol so that Alice can choose from versatile bases.

In this section, we investigate the following three scenarios, which are experimentally achievable.

\begin{enumerate}
\item Alice measures the $n$ qubits in the Pauli $X$, $Y$ or $Z$ basis, transversally.  
\item Alice measures each qubit in either the Pauli $X$ or $Z$ basis, independently.  
\\Refer to the $2^n$ measurements in total as the \emph{BB84 measurements}, as in~\cite{bennett2014quantum} in the quantum cryptography literature.  
\item Alice measures each qubit in the Pauli $X$, $Y$ or $Z$ basis, independently.  
\\Refer to the $3^n$ measurements in total as the \emph{six-state measurements}, as in~\cite{bruss1998optimal}.  
\end{enumerate}
In each of the three modified protocols, the measurement basis is chosen uniformly at random.  

It turns out that even though the test in \figref{fig:protocol} is simpler than these generalizations, at least based on a naive analysis it is as good or better for certifying dimension.  

\smallskip

The modified protocols' validities, analogous to~\thmref{thm:bound}, can also follow from the uncertainty principle. 
Indeed, one naive way to generalize the uncertainty relations to multiple orthonormal measurements is to apply~\thmref{thm:uncertainty} to all the  pairs of bases and sum them together.
Based on the corresponding uncertainty relations, bounds similar to the ones  proven in~\secref{sec:bounds} easily carry over.

\begin{corollary}
In our protocol with sufficiently large $n$, if Alice instead encodes $S$ with the Pauli $\{X,Y,Z\}$, BB84 or six-state bases, then we have, respectively,
\begin{align*} 
\text{$\{X,Y,Z\}$:} &&
\dim \H_Q&\geq 
2^{\vstretch{1.15}{(} (1 - H(\alpha)) 2 p - {\mathbf 1} \vstretch{1.15}{)} \, n - 2 H(p)}
\nonumber \\
\text{BB84:} &&
\dim \H_Q&\geq 
2^{\vstretch{1.15}{(} (1 - H(\alpha)) 2 p - \mathbf{\frac32} \vstretch{1.15}{)} \, n - 2 H(p)}   
\\
\text{Six-state:} &&
\dim \H_Q&\geq 
2^{\vstretch{1.15}{(} (1 - H(\alpha)) 2 p - \mathbf{\frac43} \vstretch{1.15}{)} \, n - 2 H(p)}  \enspace . \nonumber
\end{align*}
\end{corollary}

\begin{proof}
Iteratively applying~\thmref{thm:uncertainty} to all the basis pairs  in the Pauli $\{X,Y,Z\}$, BB84 or six-state bases and summing up, we have, respectively,
\begin{align*}
\frac{1}{3}\sum_{\theta\in\{X,Y,Z\}}H(S\vert B, \Theta=\theta)&\ge\frac12\cdot n +\frac12\,H(A|B) \,,
\\
\frac{1}{2^n}\sum_{\theta\in\, \text{BB84}}H(S\vert B, \Theta=\theta)&\ge\frac{2^{n-2}}{2^n-1}\cdot n +\frac12\,H(A|B) \,,
\\
\frac{1}{3^n}\sum_{\theta\in\, \text{six-state}}H(S\vert B, \Theta=\theta)&\ge\frac{3^{n-1}}{3^n-1}\cdot n +\frac12\,H(A|B) \,.
\end{align*}
The corollary follows by replacing the conditional entropies in the proofs of \lemref{lem:data} and~\thmref{thm:bound} with corresponding entropies in the above three inequalities.
\end{proof}

In~\cite{berta2014extractor}, Berta et al.\ derive uncertainty relations for product measurements of any full set of mutually unbiased bases (MUB) on qudits, based on the quantum-to-classical randomness extractors they construct. For qudit with prime power dimension $d$, a full set of $d+1$ MUB satisfy that for any MUB pair $\{|x\rangle\}$ and $\{|z\rangle\}$, $|\langle x|z\rangle|^2=\frac1d,\forall x,z$. For example, in the case of qubit, a full set of MUB are the Pauli $\{X,Y,Z\}$, whose product are the six-state bases.

\begin{theorem}[{\cite[Thm.~IV.4]{berta2014extractor}}]
If Alice measures each of the $n$ qudits with a random basis from a full set of MUB independently, then we have 
\beq \nonumber
\frac{1}{(d+1)^n}\sum_{\theta\in\{X^0,X^1,\ldots,X^d\}^n}H(S\vert B, \Theta=\theta)\ge\log\frac{d+1}{2}\cdot n +\min\{0,H(A|B)\} \enspace .
\eeq
\end{theorem}

\begin{corollary}
In our protocol with sufficiently large $n$, if Alice instead encodes $S\in[d]^n$ with a random basis from the product of a full set of MUB, then  we have
\beq\label{eqn:bfw} \nonumber
\dim \H_Q \geq
2^{\vstretch{1.3}{(} (\log d - H(\alpha) - \alpha \log(d-1) ) p + \log\frac{d+1}{2d} \vstretch{1.3}{)} \, n - H(p)}  \enspace .
\eeq
\end{corollary}

For comparison between different choices of encoding bases, we calculate the asymptotic number of qubits that can be certified through our lower bounds, derived from either~\thmref{thm:uncertainty} or~\cite{berta2014extractor}; see~\tabref{tab:asymptotic}.  
It is important to note that here we assume the ideal noiseless scenario where Bob is honest. That is, $\dim \H_Q = 2^n$ and Bob measures the $n$ qubits with the same basis as Alice (thus $p=1$); we also set $\alpha = O(1/n)$. 

The last three protocols in~\tabref{tab:asymptotic} are not ``complete," in the sense that the corresponding uncertainty relations fail to certify all the $n$ qubits that an honest Bob possesses.  
This is what one would expect, since the bases are not mutually unbiased, i.e., they have complementarities less than~$n$.  
Also note that for product of full set MUB on $n$ qudits, the uncertainty relations given by~\cite{berta2014extractor} are always tighter than those simply derived from~\thmref{thm:uncertainty}, except for the six-state bases on qubits, where $\log\frac32 < \frac23$.  

\begin{table}
\begin{center}
\renewcommand{\arraystretch}{1.3}
\begin{tabular}{c|c|c|c|c}
   $\{X,Z \}$& $\{X,Y,Z\}$ & BB84 & full set MUB & full set MUB\\
\hline
$n$ & $n$ & $n/2$ & $\frac{d}{d+1}\cdot n$ & $\log\big(\frac{d+1}{2}\big) \cdot n$   \cite{berta2014extractor}\\
\end{tabular}
\end{center}
\caption{Asymptotic number of certified qubits for different measurement bases, where Bob adopts the noiseless and optimal strategy, i.e., measures in the same basis as Alice. All lower bounds are proved via the two-measurement uncertainty relation in~\thmref{thm:uncertainty}, except for the case of full set of MUB on $n$ qudits, where the uncertainty relation based on certain classical-quantum extractor~\cite{berta2014extractor} is considered as well.}
\label{tab:asymptotic}
\end{table}

\bibliographystyle{alpha-eprint}
\bibliography{q}
\end{document}